\documentclass[12pt]{article}
\RequirePackage{fix-cm}
\usepackage{hyperref, mathptmx, latexsym, authblk, times}
\usepackage{mathptmx, amsmath, epsfig,graphicx,rotating}
\usepackage{dsfont,stmaryrd, bbm, wasysym, pxfonts, txfonts, amssymb, epstopdf, epsfig}
\DeclareGraphicsExtensions{.eps}

\newtheorem{theorem}{Theorem}
\newtheorem{proposition}{Proposition}

\newtheorem{definition}[theorem]{Definition}
\newtheorem{lemma}[theorem]{Lemma}
\newenvironment{proof}{\trivlist\item[]\textbf{Proof.\ }}
                      {\endtrivlist}

\begin{document}

\title{Discretization of Continuous Time Discrete Scale Invariant Processes: Estimation and Spectra}

\author{Saeid Rezakhah\footnote{ Corresponding author. Faculty of Mathematics and Computer Science, Amirkabir University of Technology, Tehran, Iran}, Yasaman Maleki \footnote{Faculty of Mathematics, Alzahra University, Tehran, Iran} }

\maketitle

\begin{abstract}
Imposing some flexible sampling scheme  we provide  some discretization of  continuous time discrete scale invariant (DSI) processes which is a subsidiary discrete time DSI process.
  Then by introducing some simple random measure we provide a second continuous time DSI process which provides a proper approximation  of the first one.  This enables us to
provide a bilateral relation between   covariance functions of the subsidiary process and the new continuous time processes.
The time varying spectral representation of such continuous time DSI  process is characterized, and its spectrum is estimated. Also, a new method for estimation time dependent Hurst parameter of such processes is provided which gives a more accurate estimation.  The performance of this estimation method is studied via simulation. Finally this method is applied to the real data of  S$\&$P500 and Dow Jones indices for some special periods.\\


\end{abstract}

\section{Introduction}
\label{introduction}
Scale invariance (or self-similarity), as an important feature, often are used as a fundamental property to interpret many natural and man-made phenomena like turbulence of fluids, textures in geophysics, telecommunications of network traffic, image processing, finance, $\cdots$ \cite{fla1}. Scale invariance is often described as a symmetry of the system relatively to a transformation of a scale, that is mainly a dilation or a contraction (up to some re-normalization) of the system parameters \cite{fla1}. Discrete scale invariance (DSI) is a property which requires invariance by dilation for certain preferred scaling factors \cite{sornett}. Burnecki et al. \cite{bur} and Borgnat et. al. \cite{fla2} have studied the property of DSI and its relation to PC by means of the Lamperti transformation \cite{characterization1}.

Even though that the re-normalization factor for scale invariant processes has some fixed exponent called Hurst parameter, but for many real world data we are dealt with  situations in which th Hurst parameter $H$, changes in time.  As part of the  literatures
 in studying   scale invariant processes with time dependent Hurst parameter $H(t)$,
 we refer the reader to the works of Flandrin and Goncalves \cite{fla3}, \cite{fla4}, Cavanaugh et al \cite{cav}, Courjouly \cite{cor}, Goncalves and Abry \cite{gon}, Kent and Wood \cite{kent}, Stoev et al \cite{sto} and Wang et al \cite{wang} s.
We present  a new method for estimating  time dependent  Hurst parameter and  implement it for simulated data and also S$\&$P500 and Dow Jones indices for some  periods.

In this paper we are to extend the work of Modarresi and Rezakhah \cite{pcls}  that they provided covariance structure and time varying
  spectral representation of certain periodically correlated processes. Even though one could provide the results expressed here in this paper by applying
  some quasi Lamperti transform, but we are to study this by an alternative way.
Modarresi and Rezakhah \cite{new1} considered a  flexible  sampling scheme  for  continuous time DSI process   and provide a discretization of this process as a discrete time DSI process.
We follow their method to introduce simple random measures on  subintervals between  sampling points of such discrete time   process,  which  amounts to the size of the process at the end point of each subinterval.
     So following their method  we provide a second continuous time DSI process which is a proper approximation  of the first one, which enables us
 to obtain  the spectral representation and spectrum  of this new continuous time DSI process.
  The covariance structure and the spectral representation of such a continuous time DSI process are specified.

Let $\{X(t), t\geq 1\}$  to be a continuous time DSI process with scale $\lambda>1$, which can be considered as the
 accumulated of flow in time. We assume that X(t) in turn can be approximated via a DSI
sequence $X^d_j$, which represents the increment of the flow on some proper partition of
time.
Following the flexible sampling scheme, we consider a partition of the positive real line as $B_j=(\lambda^{j-1}, \lambda^{j}]$ for $j \in \mathbbm{N}$ and $\lambda>1$, so  $|B_j|=\lambda^{-k}|B_{j+k}|, \; k \in \mathbbm{N}$. Then  we consider sampling of  $X(t)$  at $q$ arbitrary points $s_0 < s_1< \cdots < s_{q-1} $ inside the first scale interval $(1, \lambda]$, and follow  sampling at corresponding points $\{\lambda^{j-1} s_i:\, j \in \mathbbm{N}, i=0,1, \cdots, q-1\}$ in the successive  scale intervals \cite{characterization1}.  This provides a partition for each scale interval and divide  the scale interval $B_j$ to subintervals
 $B_{j,i} = (\lambda^{j-1}s_{i-1}, \lambda^{j-1}s_{i}], i =0,\cdots, q-1$.    Here we  introduce some certain continuous time DSI processes and provide corresponding spectral representation, which can be applied  to  approximate many  continuous time DSI processes.
In this order we extend the method of Modarresi and Rezakhah \cite{pcls}, which presented to provide spectral representation of certain continuous time  periodically correlated processes.

\subsection{Continuous time DSI process.}
Let
  $X^d_{(j-1)q+i}:=X(\lambda^{j-1}s_i )$, for $j \in \mathbb{N}$ and $i=0,\cdots ,q-1$, where $X(t)$ is some continuous time $H-\lambda$ DSI process.
Then  $X^d_j:=X(\lambda^{[j/q]}s_{j-q[j/q]})$, for $j\in \mathbb{N}$  is a subsidiary discrete time (H-$\lambda$)-DSI process.
  Now let
$M_{j}(\cdot)$  be a sequence of random measures defined on $B_j=\cup_{i=0}^{q-1}B_{j,i}$,  $ j \in \mathbbm{N},\; i=0,\cdots q-1$ where
  $  M_{j}(B_{j,i}) := X_{(j-1)q+i}^d$,  and
  $ M_{j,i}(\lambda^{j-1}s_{i-1},t]:=X_{(j-1)q+i}(t)$  for $t\in B_{j,i}$.
The covariance structure of such random measures is assigned through the covariance structure
of the discrete time DSI process $X_{(j-1)q+i}^d$.  In addition, as the amount of the process
on a set A, say the ruin due to an earthquake or a tsunami, we define simple
random measure
 $M(A)= \sum_{j=1}^m \sum_{i=1}^q M_{i}(A \cap B_{j,i})$,  where $A=\cup _{j=1}^m \cup _{i=1}^q (A \cap B_{j,i})$.

Then we study $X^d(t) = \sum_{j=1}^\infty \sum_{i=1}^q \Delta_{j,i}(t) I_{B_{j,i}}(t)$,  where $\Delta_{j,i}(t)$ is a linear combination of $X^d_{(j-1)q+i-1}$ and $X^d_{(j-1)q+i}(t)$, for $t \in B_{j,i}$  which describes elimination of the effect of
the process from the previous subinterval,say $X^d_{(j-1)q+i-1}$,   and aggregate the amount of the process on the current subinterval
 , say $X^d_{(j-1)q+i}$, as $t$ approaches to the right end point of $B_{j,i}$; say surges in a tsunami.
We also study the spectral representation of this process.

The paper is organized as follows. In Section 2, we give some background on DSI and PC processes and the Lamperti transformation. In Section 3, the main results as the covariance structure and the spectral representation of such continuous certain DSI processes are presented. In Section 4, a new method for the estimation of $H(t)$ in the DSI processes is presented and the performance and accuracy of the method is studied via simulation. The proposed estimation method is also compared with the method introduced by Modarresi and Rezakhah \cite{characterization1}. Finally, the estimation method is applied on S$\&$P500 and Dow Jones indices for some special periods.

\section{Preliminaries}
In this section, we review the spectral representation of PC processes based on unitary operators, and also the quasi Lamperti transform which connects PC and DSI processes.

\subsection{DSI and PC processes and Lamperti transformation}
A unitary operator on a Hilbert space $\mathcal{H}$, is a linear operator from $\mathcal{H}$ to $\mathcal{H}$ which preserve inner product as $<Ux, Uy>=<x, y>=
Cov(x, y) $ for every $x, y \in \mathcal{H}$, \cite{hurd}.

\begin{theorem}\label{theor}
For any unitary operator $U$ on a Hilbert space $\mathcal{H}$, there exist a unique spectral measure $Q$ on the Borel subsets of $[0, 2\pi)$ such that
$U=  \int_0^{2\pi} e^{i\omega} Q(d\omega)$, and $U^t=\int_0^{2\pi} e^{i\omega t} Q(d\omega)$ for any integer $t$,  \cite{hurd}.
\end{theorem}

\begin{proposition}\label{prop1}
A second order stochastic sequence $Y_j$ is PC with period $T$ if and only if for every $j \in \mathbbm{Z}$, there exist a unitary operator  $U=V^T$ and a periodic sequence $P_j$ with period $T$, taking values in
$\mathcal{H}_Y = \overline{sp} \{Y_j, j \in \mathbbm{Z}\}$ for which $Y_j =V^j P_j$, where $V=\int_0^{2\pi} e^{i\omega/T} Q(d\omega)$ and $Q$ is the spectral measure defined in Theorem \ref{theor}, so $Y_j=\int_0^{2\pi} e^{i\omega j/T} Q(d\omega)P_j$, \cite{hurd}.
\end{proposition}

\begin{definition}\label{dsidef}
A random process $\{X(k), k\in \widetilde{T}\}$ is called discrete time DSI in  wide sense with index $H > 0$ and scale $\lambda > 0$ with parameter space $\widetilde{T}$, where $\widetilde{T}$ is any subset of distinct points of positive real numbers, if for all $t, u \in \widetilde{T}$ and $\lambda > 0$, where $\lambda  t, \lambda u \in \widetilde{T}$:\\
$(i) \; E[X^2(t)] < \infty,$\\
$(ii) \; E[X(\lambda t)] = \lambda^{H}E[X(t)],$\\
$(iii) \; E[X(\lambda t)X(\lambda u)] = \lambda^{2H}E[X(t)X(u)].$\\
See \cite{spec}.
\end{definition}

\begin{definition}\label{lamperti3}
Let $\{Y (t), t \in \mathbbm{R}\}$ be a random process. The quasi Lamperti transform with positive parameter H and $\alpha \in \mathbbm{R}^+$, denoted by $L_{H,\alpha}$, is defined as $L_{H,\alpha}Y (t) := t^H Y (log_{\alpha} t)$,
and the inverse quasi Lamperti transform for any random process $\{X(t), t  \in \mathbbm{R}^+\}$ is
$L_{H, \alpha}^{-1}X(t) :=\alpha^{-tH}X(\alpha^ t)$, \cite{dt}.
\end{definition}

\begin{proposition}\label{propdsi}
If $\{X(t), t \in \mathbbm{R}^+\}$ is a DSI process with parameter $H$ and scale $\lambda$, then for $t\in \mathbbm{R}$, $Y (t) :=L_{H, \alpha}^{-1}X(t)$ is PC with period $T = \log_{\alpha}\lambda$. Conversely if $\{Y (t), t \in \mathbbm{R}\}$ is PC with period T,  then for $t\in \mathbbm{R}^+$, $X(t) := L_{H,\alpha}Y (t)$ is a DSI  process with parameter $H$ and scale $\lambda$, \cite{dt}.
\end{proposition}

\section{Main results}
Here, we study a new class of continuous time DSI processes $\{X^d(t), t \geq 1\}$, and investigate its special structure in subsection 3.1. Further, the covariance  structure and the spectral representation of the process are studied in subsections 3.2 and 3.3, respectively.

\subsection{Continuous DSI process}
We consider a second order discrete time DSI process $\{X^d_{(j-1)q+i}, j \in \mathbbm{N}, i=1,\cdots, q\}$ with scale $\lambda$, in the sense that  $X_{(j-1)q+i}^d \equiv \lambda^{(j-1)H}X^d_i$, and a random measure $M_{i}$ on Borel field of subsets of $B_{j,i}$, where $M_{i}(B_{j,i}):=X^d_{(j-1)q+i}$. For $A, B \subset B_{j,i}$, $C \subset B_{k,l}$, we define
\begin{equation}\label{e}
E[M_{i}(A)]=\frac{|A|}{|B_{j,i}|}E[X^d_{(j-1)q+i}],
\end{equation}
and covariance functions
\begin{equation}\label{3.2}
\sigma_{j_i,j_i}(A,B)=\frac{\beta |A| |B| + (1-\beta)m_{j,i} |A \cap B|}{m_{j,i}^2}\sigma_{j_i j_i}^d, \;\; \sigma_{j_i,k_l}(A,C):=\frac{|A||C|}{m_{j,i} m_{k,l}}\sigma_{j_i k_l}^d,
\end{equation}
where
$$\sigma_{j_i,j_i}(A,B) = <M_{i}(A), M_{{i}}(B)>, \sigma_{j_i,k_l}(A,C)=<M_{i}(A), M_{l}(C)>, \sigma_{j_i j_i}^d= Var(X^d_{(j-1)q+i}),$$
$$\sigma_{j_i k_l}^d=<X^d_{(j-1)q+i}, X^d_{(k-1)q+l}>\; =\lambda^{(j+k-2)H}< X^d_i, X^d_l>, \quad
 0\leq\beta\leq 1,$$
and $m_{j,i}= |B_{j,i}|= \lambda^{j-1}(s_i-s_{i-1})$. Also, for $j \neq k$,
$<M_{i}(A), M_{l}(B_{k,l})> = \frac{|A|}{m_{j,i}}\sigma_{j_i k_l}^d.$

Let $X^d_{(j-1)q+i}(t):=M_{i}(\lambda^{j-1}s_{i-1}, t]$, where $t=\lambda^{j-1}s_{i^*}\in B_{j,i}$, $s_{i-1} < s_{i^*} \leq s_i$.
 Thus, by Eq (\ref{e}),
\begin{equation}\label{e2}
E[X^d_{(j-1)q+i}(t)]= {\bf a}_{s_{i^*},s_i} E[X^d_{(j-1)q+i}],
\end{equation}
where  ${\bf a}_{s_{i^*},s_i}= \frac{s_{i^*}-s_{i-1}}{s_{i}-s_{i-1}}$. Moreover, for $t, u \in B_{j,i}, t \leq u$, and $v \in B_{k,l}$, by (\ref{3.2}),
\begin{equation}\label{cov}
<X^d_{(j-1)q+i}(t), X^d_{(j-1)q+i}(u)> = [\beta {\bf a}_{s_{i^*},s_i} {\bf a}_{s_{\hat{i}},s_i} + (1-\beta){\bf a}_{s_{i^*},s_{i}}]\sigma_{j_i j_i}^d,
\end{equation}
$$<X^d_{(j-1)q+i}(t), X^d_{(k-1)q+l}(v)> = {\bf a}_{s_{i^*},s_{i}} {\bf a}_{s_{l^*},s_{l}} \sigma^d_{j_i k_l}\qquad\qquad\qquad\;\;\;\;$$
\begin{equation}\label{cov2}
\qquad\qquad\qquad\qquad\qquad\quad= {\bf a}_{s_{i^*},s_{i}} {\bf a}_{s_{l^*},s_{l}} \lambda^{(j+k-2)H} < X^d_i, X^d_l>,\;
\end{equation}
where $t=\lambda^{j-1}s_{i^*}, u=\lambda^{j-1}s_{\hat{i}}, v=\lambda^{k-1}s_{l^*}$, for $ s_{i-1}< s_{i^*}\leq s_{\hat{i}} \leq s_i$, and $s_{l-1}<s_{l^*}\leq s_l$.
Consequently, if $s, w \in  B_{j_i+nq}$, $z \in B_{k_l+nq}$, $n\in \mathbbm{N}$, such that $s=\lambda^n t, w=\lambda^n u, z=\lambda^n v$, then
\begin{equation}\label{3}
<X^d_{(j+n-1)q+i}(s), X^d_{(k+n-1)q+l}(z)> = \lambda^{2nH}<X^d_{(j-1)q+i}(t), X^d_{(k-1)q+l}( v)>.
\end{equation}
 Further, by (\ref{e}), (\ref{e2}) and the DSI property of $X^d_{(j-1)q+i}$, we have that
\begin{equation}\label{e3}
E[X^d_{(j+n-1)q+i}(s)] = \lambda^{n H} E[X^d_{(j-1)q+i}(t)].
\end{equation}
 Therefore, by defining
\begin{equation}\label{dj}
N_{j,i}(t-\lambda^{j-1}s_{i-1}):= M_{i}(\lambda^{j-1}s_{i-1}, t] = X^d_{(j-1)q+i}(t),
\end{equation}
where $ t \in B_{j,i}$, we find that $N_{j,i}(y), 0< y \leq \lambda^{j-1}(s_i-s_{i-1})$ is a discrete time DSI process with scale $\lambda$ with respect to $j$, $y$, for fixed $i$.

Moreover, for $t\geq 1$ we define
\begin{equation}\label{dsi}
X^d(t) = \sum_{j=1}^\infty \sum_{i=1}^q \Delta_{j,i}(t) I_{B_{j,i}}(t); \qquad \Delta_{j,i}(t)= \frac{m_{j,i}-m_{j,i}^t}{m_{j,i}}X^d_{(j-1)q+i-1}+X^d_{(j-1)q+i}(t),
\end{equation}
where $m_{j,i}^t := t-\lambda^{j-1}s_{i-1}$. Thus, for $t=\lambda^{j-1}s_{i^*} \in B_{j,i}$,
\begin{equation}\label{dsi1}
X^d(t) = \overline{{\bf a}}_{s_{i^*},s_i}X^d_{(j-1)q+i-1}+X^d_{(j-1)q+i}(t),\;\;
\end{equation}
where $\overline{{\bf a}}_{s_{i^*},s_i}= 1-{\bf a}_{s_{i^*},s_i}$. In the next subsection, we study about the covariance structure of the certain continuous time DSI process $\{X^d(t), t \geq 1\}$.

\subsection{Covariance function of the certain DSI precess}
Now, we present a lemma which provides the covariance function of the certain continuous time DSI process $\{X^d(t), t\geq 1\}$.

\begin{lemma}\label{lem}
Let $B_{j,i}, j \in \mathbbm{N}, i=1, \cdots, q$, be a partition of the positive real line, and $\sigma_{j_i k_l}^d=<X^d_{(j-1)q+i},X^d_{(k-1)q+l}>$. The covariance function of $X^d(t)$ for $t \in B_{j,i}$ and $u\in B_{k,l}$, $t\leq u$ is
\begin{equation}\label{covm}
\sigma^d(t,u)\equiv <X^d(t), X^d(u)>\qquad \qquad \qquad \qquad \qquad \qquad \qquad \qquad\qquad \qquad \qquad \qquad \qquad \qquad
\end{equation}

$$ = \left\{\begin{array}{cc} \hspace{-.2in} & \overline{{\bf a}}_{s_{i^*},s_i} (\overline{{\bf a}}_{s_{\hat{i}},s_i} \sigma^d_{j_{i-1} j_{i-1}} + {\bf a}_{s_{\hat{i}},s_i} \sigma^d_{j_{i-1} j_i}) + {\bf a}_{s_{i^*},s_i} \overline{{\bf a}}_{s_{\hat{i}},s_i} \sigma^d_{j_i j_{i-1}} + B({s_{i^*}},{s_{\hat{i}}})\sigma^d_{j_i j_i},\;\;  t, u \in B_{j,i}, \\ \hspace{-.2in} &\\ \hspace{-.2in} & \overline{{\bf a}}_{s_{i^*},s_i} (\overline{{\bf a}}_{s_{l^*},s_l} \sigma^d_{j_{i-1} k_{l-1}} + {\bf a}_{s_{l^*},s_l} \sigma^d_{j_{i-1} k_l}) + {\bf a}_{s_{i^*},s_i} (\overline{{\bf a}}_{s_{l^*},s_l} \sigma^d_{j_i k_{l-1}} +  {\bf a}_{{l^*},s_l} \sigma^d_{j_i k_l}),\; \end{array} \right. $$

$$\qquad\qquad\qquad\qquad\qquad\qquad\qquad\qquad\qquad\qquad\qquad\qquad\qquad\qquad t \in B_{j,i} , u \in B_{k,l},$$
where for $t , u \in B_{j,i}$, $t=\lambda^{j-1}s_{i^*}$ and $u=\lambda^{j-1}s_{\hat{i}}, \; s_{i-1} < s_{i^*} \leq s_{\hat{i}} \leq s_i$, and for $u \in B_{k,l}\;$, $u=\lambda^{k-1}s_{l^*}$, $s_{l-1} < s_{l^*} \leq s_l$; also $B({s_{i^*}},{s_{\hat{i}}})= \beta  {\bf a}_{s_{i^*},s_i} {\bf a}_{s_{\hat{i}}, s_i} + (1-\beta){\bf a}_{s_{i^*},s_i}$.
\end{lemma}

\begin{proof}
By (\ref{dsi1}),
$$\sigma^d(t,u)\qquad\qquad\qquad\qquad\qquad\qquad\qquad\qquad\qquad\qquad\qquad\qquad\qquad\qquad\qquad\;$$
$$=\left\{ \begin{array}{cc} &\hspace{-.2in}  <\overline{{\bf a}}_{s_{i^*},s_i}X^d_{(j-1)q+i-1}+X^d_{(j-1)q+i}(t),  \overline{{\bf a}}_{s_{\hat{i}},s_i}X^d_{(j-1)q+i-1}+X^d_{(j-1)q+i}(u)> \qquad\; t, u \in B_{j,i}\\& \hspace{-.2in}  \\ & \hspace{-.2in}
 <\overline{{\bf a}}_{s_{i^*},s_i}X^d_{(j-1)q+i-1}+X^d_{(j-1)q+i}(t),  \overline{{\bf a}}_{s_{l^*},s_l}X^d_{(k-1)q+l-1}+X^d_{(k-1)q+l}(u)>,  t \in B_{j,i} , u \in B_{k,l}\end{array} \right.$$
 Then, by (\ref{cov}), (\ref{cov2}) and the fact that $<X^d_{(j-1)q+i-1}, X^d_{(k-1)q+l}(u)> = {\bf a}_{s_{l^*}} \sigma^d_{j_{i-1}  k_l}$, we have the result.
\end{proof}
Consequently, for $t \in B_{j,i}, u \in B_{k,l}$,  by (\ref{3}) and (\ref{covm}) we have that
$$\sigma^d(\lambda^n t, \lambda^n u) = \lambda^{2nH} \sigma^d(t, u),$$
where $\lambda^n t \in B_{j_i+nq}, \lambda^n u \in B_{k_l+nq}$.  Also  by (\ref{e3}), (\ref{dsi}),
$E[X^d(\lambda^n t)] = \lambda^{nH} E[X^d(t)],$
which shows that  $\{X^d(t), t \geq 1\}$ is a continuous time DSI process, in wide sense, with scale $\lambda$ in $t$.

\subsection{Spectral representation}
In this section we study the spectral representation structure of the continuous time DSI process $\{X^d(t), t \geq 1\}$.

Let $\{X^d_{(j-1)q+i}, j \in \mathbbm{N}, i=1, \cdots, q\}$ be a subsidiary DSI process with scale $\lambda$ in $j$ for fixed $i$, in the sense that $X^d_{(j-1)q+i} \equiv \lambda^{(j-1)H}X^d_i$.
Modarresi and Rezakhah \cite{new1} provided the spectral representation and spectral density matrix  of the corresponding
subsidiary multi-dimensional self-similar corresponding to such subsidiary DSI process which can be applied to obtain the spectral representation and spectral density function of such continuous time DSI process.
Also Modarresi and Rezakhah \cite{spec} provided the spectral representation of the corresponding discrete time DSI process when $\lambda=\alpha^q$ and $s_j=\alpha^J$ which can be applied to obtain spectral representation of the process in such a case.

\vspace{0.1in}\label{mse3}
\input{epsf}
\epsfxsize=1.5in \epsfysize=1.5in
\begin{figure}\label{mse2}
\vspace{-0.21 in}
\centerline{\includegraphics[scale=0.47]{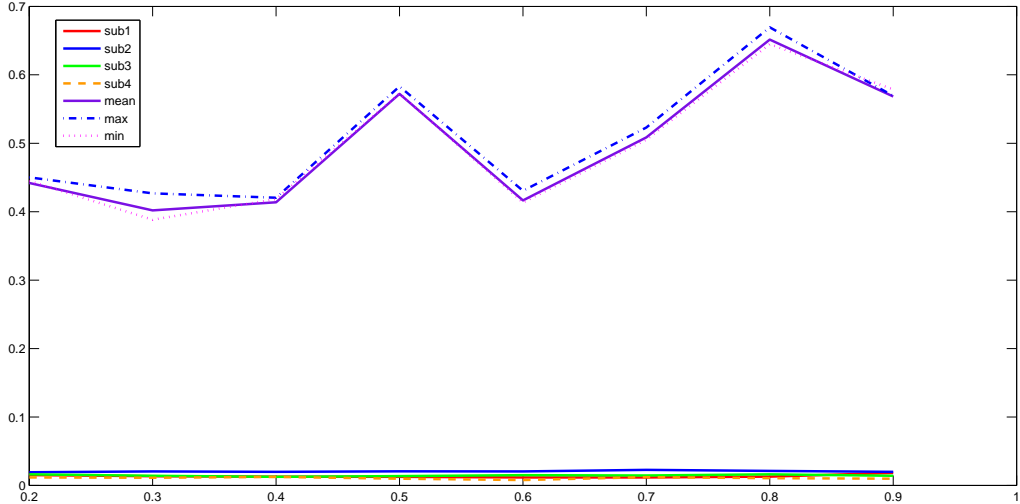}}
\vspace{-0.1in}
\caption{\scriptsize MSEs of Hurst estimator in 100 repetitions of DSI processes with 4 scale intervals, 4 subintervals and 80 equally spaced samples points in each scale interval. MSEs of $\hat{H} = (\hat{H_1}, \hat{H_2}, \hat{H_3}, \hat{H_4})$ are indicated with different colors for subintervals 1, 2, 3, 4. Also, MSEs of the estimator proposed by Modarresi and Rezakhah \cite{characterization1} are indicated for mean, maximum and minimum of $\hat{H}$ for different values of $H$.}
\end{figure}
\vspace{-0.1in}
\section{Simulation}
This section, we present a new method for the estimation of Hurst parameter of  DSI processes which is more compatible for real world data. 
This motivation comes from the fact that in many DSI processes we have different behavior for the fluctuation of the process in different parts of scale intervals which are followed in corresponding parts of successive scale intervals. So assuming the Hurst parameter $H$ to be some fixed number 
 for the whole process, in many real data is not realistic. The other behavior of the DSI process is some self-similar behavior inside the scale intervals. Therefore, it would be useful to consider some  partition  of scale intervals, say $q$, that the ratio of the length of each subinterval to the
  length of the corresponding scale intervals
   are fixed through out the process. Then we evaluate the estimation of the Hurst parameter corresponding to the growth of the process in corresponding 
   subintervals.  By assuming $q$ such partition for each scale interval, we have a  Hurst estimation vector of size $q$ to be estimated.
\vspace{0.25in}
\input{epsf}
\epsfxsize=2.5in \epsfysize=1.5in
\begin{figure}\label{ss2000}
\vspace{-0.1 in}
\centerline{\includegraphics[scale=0.42]{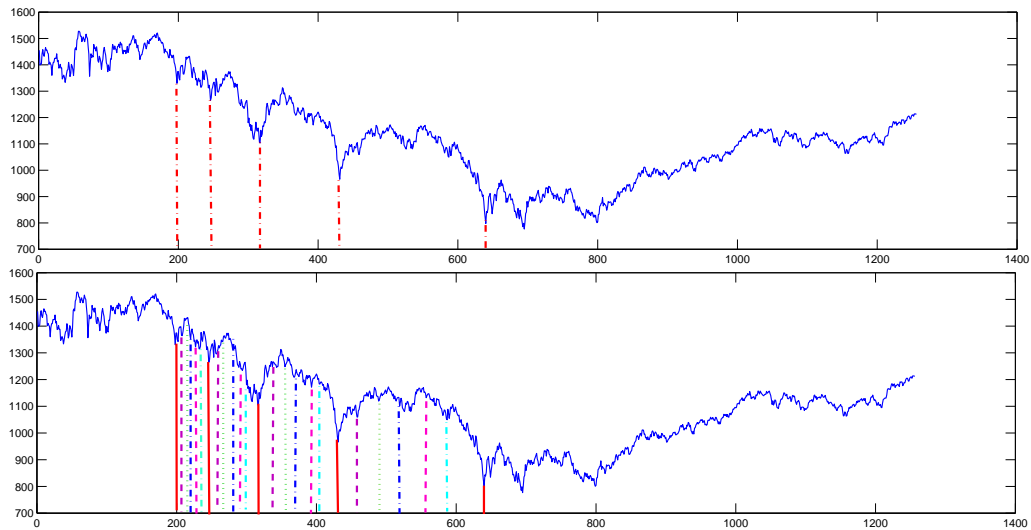}}
\vspace{0.1in}
\caption{\scriptsize S$\&$P500 indices from 1/1/2000 until 31/12/2004. The existence of a DSI
behavior is justified from 16/10/2000 until 23/7/2002 in four scale intervals which are indicated with red lines (Top). Also, subintervals are indicated with colored dashed lines (Bottom). The scale of the process for the periods is evaluated approximately with 1.66.}
\end{figure}
  Extending the method  proposed by Modarresi and Rezakhah \cite{characterization1} for scale intervals, we evaluate the ratios of quadratic variations of corresponding subintervals in consecutive scale intervals by the following and estimate the Hurst vector.
Thus we calculate the quadratic variations corresponding to the $i-$ th subinterval in the $j-$ th scale interval as
\begin{equation}\label{ss}
SS_{j,i}=\frac{1}{l_i-1}\sum_{k=2}^{l_i} (x(t_{(j-1)q+L_{i-1}+k})-x(t_{(j-1)q+L_{i-1}+k-1}))^2,
\end{equation}
\noindent for $j=1, \cdots, m$, $i=1, \cdots, q$; where $l_i$ is the number of observations in $i-$ th subinterval, and $L_i=\sum_{k=1}^i l_k$ , $L_0=0$. Then, we use
\begin{equation}\label{mu}
\mu_{j,i}=\frac{\log(SS_{j+1,i}/SS_{j,i})}{2\log \lambda}.
\end{equation}
Finally, the Hurst estimation for $i-$ th subinterval is obtained as
$\widehat{H}_i=\frac{1}{m-1}\sum_{j=1}^{m-1} \mu_{j,i}$.
\vspace{0.15in}
\input{epsf}
\epsfxsize=2.5in \epsfysize=1.5in
\begin{figure}\label{sp222}
\vspace{-0.1 in}
\centerline{\includegraphics[scale=0.42]{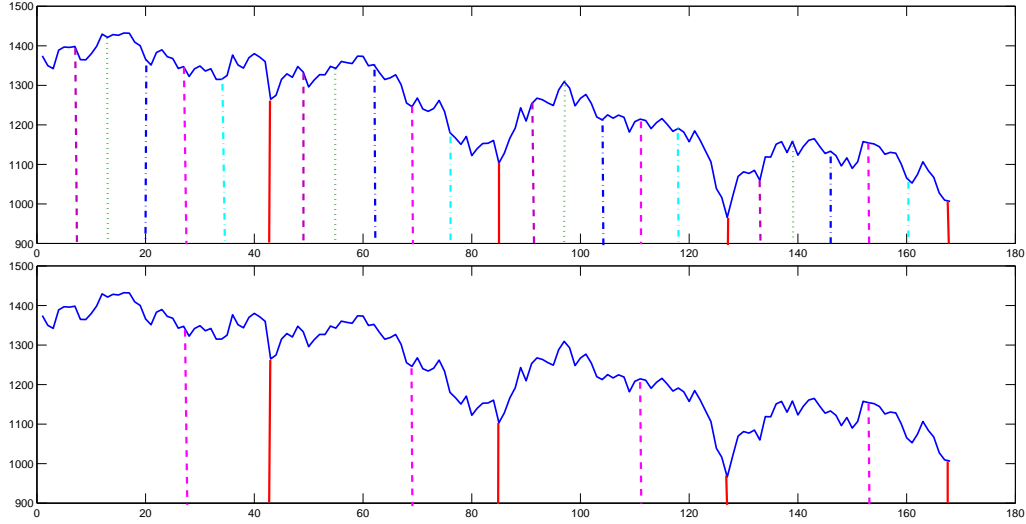}}
\vspace{0.1in}
\caption{\scriptsize The sampled S$\&$P indices from 16/10/2000 until 23/7/2002 obtained by the sampling method. Scale intervals are indicated with red lines. Also, 6 subintervals (Top), and 2 subintervals (Bottom) are indicated with colored dashed lines.}
\end{figure}
The superiority of the method is investigated for simulated and empirical data. In the simulation case, MSEs of the proposed method and the method presented in \cite{characterization1} are computed for different Hurst indices.  To this end, we consider DSI processes in which the Hurst parameter changes slowly in time, such that we can assume some subintervals in each scale interval where the Hurst parameter is approximately constant for the entire subinterval, but it changes from one subinterval to another. As an example, we consider a DSI process which contains 4 scale intervals, 4 subintervals and 80 equally spaced samples points in each scale interval. Hence, the $H$ estimator would be a vector of size 4, as $\hat{H} = (\hat{H_1}, \hat{H_2}, \hat{H_3}, \hat{H_4})$. The MSEs of $\hat{H}$ are computed and depicted in Figure 1,  with different colors for subintervals 1, 2, 3, 4, and for different values of $H$. Also, the MSEs of the estimator proposed by Modarresi and Rezakhah \cite{characterization1} are indicated for the mean, maximum and minimum of $\hat{H}$, for different values of $H$. As can be seen in the Figure 1, considering subintervals in each scale interval gives a more accurate estimation than the method presented in \cite{characterization1}.
\vspace{0.5in}
\input{epsf}
\epsfxsize=2.5in \epsfysize=1.5in
\begin{figure}\label{dow}
\vspace{-0.1 in}
\centerline{\includegraphics[scale=0.42]{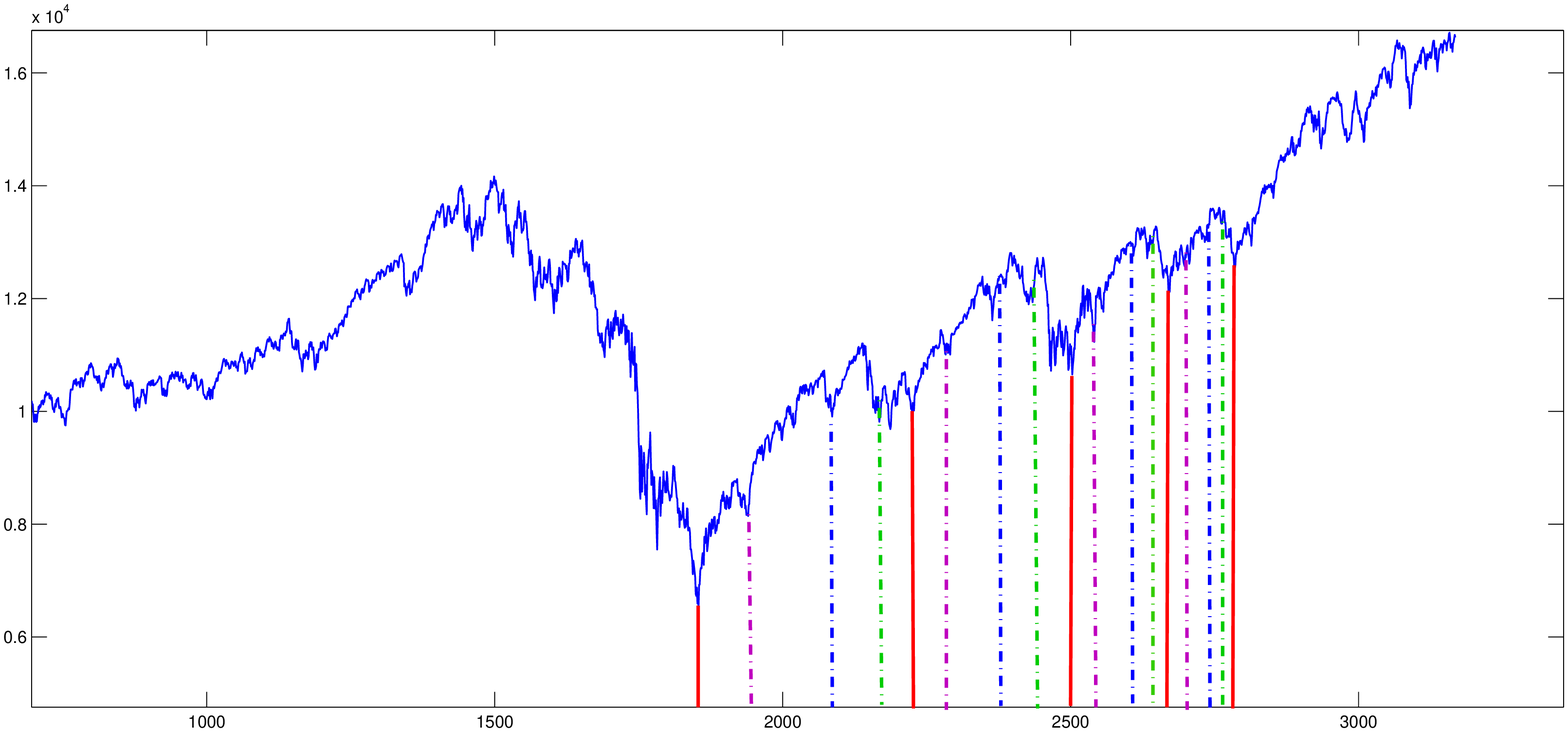}}
\vspace{-0.1in}
\caption{\scriptsize Dow Jones indices from 25/10/2001 until 28/5/2014. The existence of a DSI behavior is justified from 6/3/2009 until 14/11/2012 in four scale intervals which are indicated with red lines. Also, subintervals are indicated with colored dashed lines. The scale of the process for the periods is evaluated approximately with 1.493.}
\end{figure}

\vspace{-0.25 in}

\subsection{Empirical Data}
The superiority of the method is also investigated for empirical data. To this end, we study the daily indices of two  stock markets: S$\&$P500 and Dow Jones. First, we consider daily indices of S$\&$P500 from the first January 2000 till  the end of 2004. As there is not any index on Saturdays, Sundays and holidays, the available data for the selected period are 1256 days. The time series of these indices is shown in Figure 2. These data are also studied by Bartolozzi et al. \cite{bar}, where the existence of a DSI behavior, with scale 2, in some periods of data has been justified. We only consider the indices from 16th October 2000 until 23th July 2002 which the DSI behavior can be seen in four scale intervals and is shown in Figure 2 by red lines. The preferred scaling factor of the process for the periods is evaluated approximately with 1.66. Therefore, the end points of the scale intervals would be $b_1 = 200$, $b_2 = 246$, $b_3 = 317$, $b_4 = 431$ and $b_5 = 640$. To study these four scale intervals accurately, subintervals should be determined such that the process behaves the same in all consecutive subintervals. For the  S$\&$P500  process, the subintervals of the first scale interval are indicated with the end points: $i_{11}=200$, $b_{12}=206$, $b_{13}=212$, $b_{14}=219$,  $b_{15}=226$, and $b_{16}=233$, which are shown in Figure 2 with different colors. The corresponding subintervals of the $j-$th scale interval, j = 2, 3, 4, are determined as $b_j + [(1.66)^{j-1} * i]$, where $i = 6,12,19,26,33,41$, and  $b_j$ is the starting point of the $j-$th scale interval, also, $[.]$ is the bracket. Following the  sampling method, we consider 42 arbitrary sample points in the first scale interval as 200, 201, 202, $\cdots$, 241, and corresponding sample points in the $j-$th scale interval, can also be determined as  $b_j + [(1.66)^{j-1} * i]$, j = 2, 3, 4. So, we would have 6, 6, 7, 7, 7, 9 data in subintervals 1, 2,  $\cdots$, 6 of the first scale interval, respectively. The sampled S$\&$P500 process, obtained by the sampling method, is depicted in Figure 3.

By computing ratios of sample variances of subintervals, Eq (\ref{ss}), (\ref{mu}), the Hurst estimation comes out as $\widehat{H} = [0.24, 0.23, 0.13, 0.24, 0.07, 0.05]$. Since, the first four Hurst estimations, and the last two ones are approximately the same; so, we combine the first four subintervals as one subinterval, and the last two ones as the second subinterval. Hence, there would be 26 and 16 observations in the first and  second subintervals, respectively. Such partitions are shown in Figure 3 (Bottom). Thus, the Hurst estimations would be: $\widehat{H}=[0.23, 0.06]$. The Hurst parameter also estimated based on the method presented in  \cite{characterization1} as  $\widehat{H}=0.16$.

As an another example, we consider daily indices of Dow Jones from 25th October 2001 till 28th May 2014. As there is not any index on Saturdays, Sundays and holidays, the available data for the selected period are 3168 days. These indices are plotted in Figure 4, where the existence of a DSI behavior in a period from 6th March 2009 until 14th November 2012 has been justified in four scale intervals.
 Also, the scale $\lambda$ has been evaluated approximately with 1.493. The red lines in Figure 4 reveals the scale
 intervals with the end points $b_1 = 1853$, $b_2 = 2225$, $b_3 = 2503$, $b_4 = 2671$ and $b_5 = 2748$. To estimate the
 Hurst parameter accurately, 4 subintervals are determined in each scale interval. Since the last scale interval, is the shortest one,
 so we  determine subintervals in it with the end points  $b_{11}=2671$, $b_{12}=2698$, $b_{13}=2741$, $b_{14}=2766$, and $b_{15}=2784$.
 The subintervals of the rest scale intervals could be found from the relation $b_{5-j}+[1.493^{j-1}*i]$, where  $i = 27, 70, 95$,
 and  $b_j$ is the starting point of the $j-$th scale interval.
 Partitioning of scale intervals are shown in Figure 4 with different colors. Also, 113
 arbitrary sample points are considered in the last scale interval as 2671, 2672, 2673, $\cdots$, 2783.
 The corresponding sample points in the $j-$th scale interval, j = 1, 2, 3, can be determined as  $a_{5-j}+[1.493^{j-1}*i]$,
  where for the indices in the mentioned period, we would have 27, 43, 25, 18 data in subintervals 1, 2, 3, 4, respectively.
  Using (\ref{ss}) and (\ref{mu}), the Hurst estimation comes out as: $\widehat{H} = [0.46, 0.56, 0.63, 0.50]$,
  which shows the Hurst indices of subintervals 1-4. Also, based on the method provided by Modarresi and Rezakhah \cite{characterization1},
  the Hurt estimation for the whole process is computed as  $\widehat{H}=0.53$.
\section{Conclusion}
A certain continuous time DSI process is defined, where its covariance function is generated by the covariance function of a discrete time DSI process through defining some simple random measure on the positive real line. The  time varying spectral representation of such a process is characterized, and its spectrum is estimated. Moreover, a new method for Hurst estimation of DSI processes is proposed. The estimation  method is based on determining some subintervals in each scale interval and evaluating ratios of quadratic variations corresponding to consecutive subintervals. Numerical results were presented to show that the proposed method provided an accurate estimation than the method provided in \cite{characterization1}, in the sense of MSE. Finally, the estimation method is applied to the S$\&$P500 and Dow Jones indices for some special periods.

\bibliographystyle{unsrt}

\end{document}